\documentclass[12pt,final]{article}

\usepackage{fixme}
\fxsetup{
	status=draft,
	author=Dima,
	layout=inline, 
	theme=color
}

\usepackage{url}
\usepackage{latexsym,amssymb,lastpage}
\usepackage{graphicx,amsfonts,euscript}
\usepackage{times,mathptmx,bm,amsmath,amsthm}
\usepackage{dcolumn}
\usepackage{showkeys}
\usepackage{overpic}

\usepackage[cp1251]{inputenc}
\usepackage[english]{babel}

\usepackage{color}
\definecolor{gray10}{gray}{0.9}
\definecolor{gray20}{gray}{0.8}
\definecolor{gray30}{gray}{0.7}
\definecolor{gray40}{gray}{0.6}
\definecolor{gray50}{gray}{0.6}
\definecolor{gray60}{gray}{0.4}
\definecolor{gray70}{gray}{0.3}
\definecolor{gray80}{gray}{0.2}

\newtheorem{lemma}{Lemma}
\newtheorem{theorem}{Theorem}
\newtheorem{example}{Example}
\newtheorem{corollary}{Corollary}

\graphicspath{{../pic/}}

\def\R{\mathbb R}
\def\C{\mathbb C}

\def\Z{\mathbb Z}

\def\a{a}

\newcommand{\suml}{\sum\limits}
\newcommand{\intl}{\int\limits}

\newcommand{\e}{{\rm e}}

\DeclareMathOperator{\Arg}{Arg}

	\title{Linear nonlocal problem for the abstract time-dependent non-homogeneous Schr\"odinger equation}
\author{Dmytro Sytnyk \\
		\small Institute of mathematics, National Academy of Sciences, Ukraine,\\
		\small sytnikd@gmail.com
		\and 
		Roderick Melnik \\
		\small Wilfrid Laurier University, Waterloo, Canada,\\
		\small rmelnik@wlu.ca}

\begin{document}
\maketitle
\begin{abstract}

A nonlocal-in-time problem for the abstract Schr\"odinger equation is considered. 
By exploiting the linear nature of nonlocal condition we derive an exact representation of the solution operator 
under assumptions that the spectrum of Hamiltonian is contained in the horizontal strip of complex plane.
The derived representation permits us to establish the necessary and sufficient conditions for the problem's well-posedness and the existence of its mild, strong solutions.  
Furthermore, we present new sufficient conditions for the existence of solution which extend the available results to the case when some nonlocal parameters are unbounded.
Two examples are provided.

\end{abstract}

\section{Introduction}
In the abstract setting the evolution of quantum system is governed by differential equation 
\begin{equation}\label{SchrodEqt}
	{i}\psi'_t-H\psi= i v(t), \quad t \in [0, T], 
\end{equation}
which is called time-dependent non-homogeneous Schr\"odinger equation.
Standard axiomatic approach to the quantum mechanics ensures that the state of the system described by a wave function  $\psi(t) \in X$,  is uniquely determined by \eqref{SchrodEqt} and a given initial state $\psi_0$
\begin{equation}\label{eq:SEIC}
\psi(0) = \psi_0.
\end{equation}
This is achieved by requiring that the linear operator (Hamiltonian) $H:\ X \rightarrow X$ is self-adjoint in the Hilbert space $X$ and its domain $D\left (H\right ) \subseteq X$ is dense. 
Stone's theorem states that in such a case there exists a strongly continuous unitary group $U(t) = \e^{-iHt}$  with generator $iH$ \cite{Fattorini1985}. 
The function $\psi(t)$ is called a mild solution of \eqref{SchrodEqt}, \eqref{eq:SEIC}, if it satisfies the equation
\begin{equation}\label{eq:MildSolSchrodCP}
\psi(t)  = U(t) \psi(0) + \intl_{0}^{t}U\left(t-s\right) v(s)ds, \quad t \in [0, T]. 
\end{equation}
Substitution of the initial data from  \eqref{eq:SEIC} into this general solution representation  lead us to the usual propagator formula $\psi(t)  = U(t) \psi_0$ for  the solution of 
\eqref{SchrodEqt}, \eqref{eq:SEIC} with $v(t) \equiv 0$.

In this work we consider a nonlocal generalization of condition \eqref{eq:SEIC}:
\begin{equation}\label{eq:linear_nc}
	\psi(0)+\suml_{k=1}^{n}\alpha_k \psi(t_k) =\psi_1.
\end{equation}
For the fixed  state $\psi_1 \in X$ this condition is determined by the set of parameters $ 0<t_1<t_2<\ldots < t_n \leq T,\ \alpha_k \in \C$ which will be called the \textit{parameters of nonlocal condition}.
Aside from the standard initial condition \eqref{eq:SEIC} it generalizes other important condition types, such as periodic conditions $\psi(0) = \psi(t_1)$ and Bitsadze-Samarskii conditions $\psi(0)+\alpha_1 \psi(t_1)=\alpha_2 \psi(t_2)$ (e.q. \cite{Ashyralyev2008}).
 Formula \eqref{eq:linear_nc} can be also viewed as approximation to a more general nonlinear condition $\psi(0) + g\left (t_1, \ldots t_k, \psi\left ( \cdot \right )\right ) = 0$ for a suitably defined function $g\left (t_1, \ldots t_k, \cdot \right ):\ X \rightarrow X$.
Among other applications, nonlocal problem \eqref{SchrodEqt}, \eqref{eq:linear_nc} is important for the theory of driven quantum systems, where one is interested in a way to recapture specific nonlocal behavior of solution $\psi(0) =\alpha_1 \psi(t_1)+\psi_1$ by changing the properties of driving potential $p(t)$ from the Hamiltonian $H = H_0 + p(t)$. 
To stay within the classical formulation \eqref{SchrodEqt},\eqref{eq:SEIC} this theory operates upon assumption that $p(t)$ is periodic 
\cite{KuwaharaMoriSaito2016,LazaridesDasMoessner2015,VerdenyMintert2015}.
Then a predictable nonlocal-in-time behavior of the system follows from the Floquet theorem \cite{daners1992abstract}.
The case of non-periodic $p(t)$ is much harder to treat, since the Floquet theory can not be applied.
It is our belief that the nonlocal formulation is a viable alternative to other proposed  generalizations of periodic quantum driving that are currently under research 
\cite{VerdenyPuigMintert2016,YanLiYangEtAl2015}.  
The above mentioned two-point nonlocal condition with $\psi_1\neq0$ can also be thought of as a generalization of the renowned Rabi problem \cite{Gardas2013,LeHur2016} used in the modern quantum computing for state preparation and information processing \cite{LeBellac2006}. 
A nonlinear problem similar to \eqref{SchrodEqt},\eqref{eq:linear_nc} was analyzed in \cite{bunoiu2016vectorial}, where the authors were motivated by the study of Bose-Einstein condensates. 

In spite of the increasing importance, 
a surprisingly little is known about the solution of \eqref{SchrodEqt}, \eqref{eq:linear_nc}. 
This problem was studied in \cite{ashyralyev2008nonlocal}, using the Hilbert space methods. 
For the self-adjoint $H$ it was proved that the condition
\begin{equation*}\label{eq:NCNS_ash_sufff_cond}
	\suml_{k=1}^{n}\left |\alpha_k \right | <1
\end{equation*}
is sufficient for the existence of solution to \eqref{SchrodEqt}, \eqref{eq:linear_nc}, when $\psi_1$ is a smooth enough vector with respect to $H$ (see \cite{ashyralyev2008nonlocal} for the details). 
The same condition appeared earlier in \cite{Byszewski1992}, where a more general nonlocal problem for the first order equation with sectorial operator coefficient in a Banach space was considered. 
In the current work we focus on a following generalization of the  condition from \cite{ashyralyev2008nonlocal,Byszewski1992}: 
\begin{equation}\label{estLiang2002} 
\suml_{k=1}^{n} |\alpha_k|e^{d t_k} \leq 1,
\end{equation}
 developed in \cite{NonlocalAbsNonLinNtouyas1997}. 
 Here $d$ is a half-height of the strip containing the spectrum of operator $H$ defined in the Banach space $X$.
In the course of the work we show that inequality \eqref{estLiang2002} represents only a fraction of the parameter space where problem \eqref{SchrodEqt}, \eqref{eq:linear_nc} is well-posed and have a mild solution defined by \eqref{eq:MildSolSchrodCP}.
 More generally, we establish new necessary and sufficient conditions for the existence of solution to \eqref{SchrodEqt}, \eqref{eq:linear_nc} which can be verified for any given set of $\alpha_k,t_k$ from \eqref{eq:linear_nc}.
 In addition to that, we derive several versions of sufficient conditions for the solvability of the given nonlocal problem which extend the region of admissible $\alpha_k$ outside the manifold governed by \eqref{estLiang2002}.

The paper is organized as follows. 
In Section \ref{sec:hs_operator_calculus} we introduce a notion of strip-type  operators $H$ acting on Banach space $X$ and review the functional calculus of such operators.
Our aim is to specify the class of $H$ such that the propagator $U(t)$ is well-defined and can be represented via the Dunford-Cauchy formula.
Section \ref{sec:NCNS_reduction} is devoted to the analysis of solution existence. 
We start with the reduction of nonlocal problem \eqref{SchrodEqt}, \eqref{eq:linear_nc} to classical Cauchy problem \eqref{SchrodEqt}, \eqref{eq:SEIC}. 
Then the operator calculus of Section \ref{sec:hs_operator_calculus} is applied to study the obtained solution operator of nonlocal problem.
Theorem \ref{thm:NCNS_exist_mild} gives the necessary and sufficient conditions for the existence and uniqueness of mild solution to \eqref{SchrodEqt}, \eqref{eq:linear_nc}. 
Corollaries \ref{thm:NCNS_exist_strong} and \ref{thm:NCNS_well_posed} concern the existence of strong solution and the well-posedness of the given problem. 
Apart from several simple cases, 
the conditions on parameters $\alpha_k,t_k$, mentioned in Theorem \ref{thm:NCNS_exist_mild},  can be verified if the values of these nonlocal parameters are specified.
In Section \ref{sec:NCNS_polynomial_reduction} we further adopt the technique from \cite{nonloc_exMVS2014}, which, when suited with the properly chosen conformal mapping (adjusted to the spectral-strip parameter $d$), permits us to reduce the question of solution's existence to the question about the location of roots for a certain polynomial associated with the nonlocal condition.
This, in turn, enables us to obtain the conditions for the existence and uniqueness of the solution to  \eqref{SchrodEqt}, \eqref{eq:linear_nc} stated in terms of the constraints on $t_k,\ \alpha_k$  (Theorems \ref{thm:NCNS_suf_cond_root_est}, \ref{thm:NCNS_criteria}).
Finally, we compare newly derived conditions against \eqref{estLiang2002}, using the three-point nonlocal problem as a model example.

\section{Functional calculus of strip-type operators}\label{sec:hs_operator_calculus}
With intent to study problem \eqref{SchrodEqt}, \eqref{eq:linear_nc} in a Banach space setting, 
in this section, we review necessary facts from the 
holomorphic functional calculus for operators with the spectrum in a horizontal strip \cite{Haase2007}.
A densely defined closed linear operator $H$ with the domain $D(H) \subseteq X$, whose spectrum belongs to the set 
\begin{equation}\label{eq:SpHalfStrip}
\Sigma_d = \left\{
z=x + i y \middle|\ x,y \in \R,\ |y|\leq d 
\right\},
\end{equation} 
and the resolvent $R\left(z,H\right) \equiv (zI-H)^{-1}$ satisfies
\begin{equation}\label{eq:ResHalfStrip}
\left \|R\left(z,H\right)\right \|\leq \frac{M}{|\Im{z}|-d},\quad z \in \Omega\setminus \Sigma,\  \Sigma \subset \Omega,
\end{equation}
is called a strip-type operator of the height $2d >0$. 

Next we define the rule to interpret operator functions. 
Let $f(z)$ be a complex valued function  analytic in the neighborhood $\Omega$ of the spectrum $\Sigma(H) \subset \C$ and $|f(z)|<c_f\left (1+|z|\right )^{-1-\delta}$, for $\delta >0$. Suppose that there exists a closed set $\Phi \subset \Omega$ with the boundary $\Gamma$ consisting of a finite number of rectifying Jordan curves,  then the operator function $f(H)$ can be defined as follows 
\begin{equation}
\label{reprDunford}
f(H) x=\frac{1}{2\pi i} \intl_{\Gamma}f(z) R(z,H)x dz.
\end{equation}
This formula yields an algebra homomorphism between the mentioned class of holomorphic functions 
and the algebra of bounded operators on $X$. Besides, any two valid functions of the same operator commute.
Unfortunately, Dunford-Cauchy integral \eqref{reprDunford}  can not be used straight away to define the propagator, 
because $|\e^{-iz}|$ will not vanish as $z\rightarrow\infty$ on $\Gamma$. Assume that there exists a so-called \textit{regularizer} function $\epsilon(z)$ such, that both $e(H)$ and $ef (H)$ are well defined in terms of \eqref{reprDunford} and $e(H)$ is injective. Then the formula 
\begin{equation}\label{eq:func_calc_regularizer}
f(A) = e^{-1}(H) ef (H)
\end{equation}
is used to define $f(H)$ for a class of functions wider than the natural function calculus defined by \eqref{reprDunford} alone. 
By setting $e(z)=(\lambda - z)^{-1-\delta}$ with $|\Im{\lambda}|>d$ we ensure that $f(H)x$ is well defined and bounded, whenever $f(z)$ is bounded in $\Omega$ and $e^{-1}(H)x$ exists. 
In other words, the propagator $U(t)$ is bounded linear operator with the domain $x \in D(H^{1+\delta})$. 
By using the closed graph theorem \cite{Munkres2005}, $U(t)$, $t \in \R$ can be extended to the bounded operator on $X$ when the set $D(H^{1+\delta})$ is dense in $X$. 
For more details on the construction and properties of the  functional calculus for strip-type operator we direct the reader to \cite[Chapter 4]{haase2006functional}. 

\section{Reduction of nonlocal problem to classical Cauchy problem}\label{sec:NCNS_reduction}
We depart from the general solution formula \eqref{eq:MildSolSchrodCP}, with $\psi(0)$ supplied by \eqref{eq:linear_nc} 
\begin{equation}\label{eq:linear_NCNS}
\psi(t)=U(t)\left(\psi_1-\suml_{k=1}^{n}\alpha_k \psi(t_k) \right)+\intl_{0}^{t}U\left(t-s\right) v(s)ds,
\end{equation}
that is valid for the strip-type operator $H$ under assumptions of Section \ref{sec:hs_operator_calculus}. 
To get the exact representation for $\psi(t)$ one needs to factor out the unknown $\psi(t_k)$, $k=\overline{1,n}$ from the above formula. 
We define 
$w \equiv \suml_{k=1}^{n}\alpha_k \psi(t_k)$ and  then formally evaluate this expression using \eqref{eq:linear_NCNS} as a representation for $\psi(t)$.
It leads to the equation 
$$
w = {-\suml_{i=1}^n \alpha_i U(t_i) w} +\suml_{i=1}^n \alpha_i U(t_i) \psi_1+ \suml_{i=1}^n \alpha_i \intl_0^{t_i} U(t_i-s)v(s) ds.
$$
By denoting 
$B = I +\suml_{i=1}^n \alpha_i U(t_i)$ we rewrite this equation as follows
\begin{equation}\label{eqnoneq1}
B w = B \psi_1 -\psi_1 +\suml_{i=1}^n \alpha_i \intl_0^{t_i} U(t_i-s)v(s) ds.
\end{equation}

At this point it is clear that equation \eqref{eqnoneq1} can be solved for $w$ with any combination of $\psi_1$ and $v(t)$ if and only if the operator function  $B$ posses the inverse $B^{-1}$. 
In such case the substitution  
\begin{equation}\label{eq:w_repr}
w = \psi_1- B^{-1} \psi_1 + B^{-1}\suml_{i=1}^n \alpha_i \intl_0^{t_i} U(t_i-s)v(s) ds
\end{equation}
into \eqref{eq:linear_NCNS} yields a representation of the general (mild) solution to nonlocal problem \eqref{SchrodEqt}, \eqref{eq:linear_nc}
 \begin{equation}
 \label{bp1IntRed}
 \begin{split}
 \psi(t)=& U(t)\left(B^{-1}\psi_1 -B^{-1}\suml_{i=1}^n \alpha_i \intl_0^{t_i} U(t_i-s)v(s) ds\right)+\intl_{0}^{t}U\left(t-s\right) v(s)ds.\\
 \end{split}
 \end{equation}
Now we can formalize our previous analysis as a theorem.

\begin{theorem}\label{thm:NCNS_exist_mild}
	Let  $H$ be a strip-type operator with the spectrum $\Sigma$, having nonempty point-spectrum component, and the domain $D(H^\delta)$ is dense in $X$ for some $\delta > 1$. 
	The mild solution of nonlocal problem \eqref{SchrodEqt}, \eqref{eq:linear_nc} exists and is unique for any 
	$\psi_1 \in X$, $v \in L^1((0;T),X)$  
	if and only if all the zeros of the entire function 
	\begin{equation}
	\label{zerosExpI}
	b(z)=1+\sum_{k=1}^n{\alpha_k e^{(-i t_k z)}}, 
	\end{equation}
	associated with \eqref{eq:linear_nc}, are contained in the interior of the set $\mathbb{C} \backslash \Sigma$.
\end{theorem}

\begin{proof}
We prove necessity first. 
A solution to the given nonlocal problem satisfies differential equation \eqref{SchrodEqt}, hence general representation \eqref{eq:MildSolSchrodCP} is valid for such solution with any given combination of $\psi (0)$, $v(t)$.
Upon setting  $v(t) = 0$ in this representation, we substitute it into \eqref{eq:linear_nc} to get the equation
\begin{equation}\label{eq:nc_zero_v}
B\psi(0) = \psi_1
\end{equation}
with respect to $\psi(0)$.	 
Suppose that the function $b(z)$ has a root $z_0 \in \Sigma$ which belongs to the point spectrum of $H$, with $\varphi \neq 0$ being the corresponding eigenstate. 
Now, we pick a bounded sequence $\left\{ \psi_{1k}\right\}_{k=1}^\infty$, so that $\psi_{1k} \in D(H^\delta)$, $\psi_{1k} \neq \varphi$ and $\psi_{1k} \rightarrow \varphi$ strongly. 
Such sequence always exists since the domain $D(H^\delta)$ is dense in $X$.
By the theorem's premise, for any $\psi_1$ there should exist a corresponding bounded  state $\psi(0)$ satisfying \eqref{eq:nc_zero_v}.   
To show that this is not true for $\psi_1 = \varphi$ we, first, evaluate $B\varphi$ via the Dunford-Cauchy integral 
\[
B\varphi =\frac{1}{2\pi i}\intl_{\Gamma} b\left (z\right )R(z,H) \varphi dz = \frac{1}{2\pi i}\intl_{\Gamma} \frac{b\left (z\right )}{z-z_0} \varphi  dz = b(z_0) = 0, 
\]
and then apply the general inequality $\|B^{-1}\| \geq \frac{1}{\|B\|}$ to $B^{-1}\psi_{1k}$:
$$
\lim\limits_{k \rightarrow \infty}\|B^{-1}\psi_{1k}\| \geq \lim\limits_{k \rightarrow \infty}\frac{1}{\|B\psi_{1k}\|} = \infty .
$$

Next we prove sufficiency.  
Assume that all the zeros of $b(z)$ belong to the interior of $\mathbb{C} \backslash \Sigma$. 
By using the operator function calculus from Section \ref{sec:hs_operator_calculus} we define
\begin{equation}
\label{reprDelta}
B^{-1}\varphi =\frac{1}{2\pi i}\intl_{\Gamma} \frac{1}{b\left (z\right )} R(z,H) \varphi dz, 
\end{equation}
for any 
$\varphi \in X$.
The contour $\Gamma$ satisfying the requirements of \eqref{reprDunford} exists, since $1/b(z)$ is holomorphic in the neighborhood of $\Sigma$. 
Formula \eqref{reprDelta}, the condition $v \in L^1((0;T),X)$ and Lemma 5.2 from \cite{Fattorini1985} guaranty that the state $\psi_0$ given by 
\begin{equation}\label{eq:initial_state_NCNS}
\psi_0 = B^{-1}\psi_1 -B^{-1}\suml_{i=1}^n \alpha_i \intl_0^{t_i} U(t_i-s)v(s) ds,
\end{equation}
is well-defined for any combination of $v(t)$ and  $\psi_1$ fulfilling  the theorem's assumptions.
That, in turn,  implies a well-definiteness of $\psi(t)$ given by formula \eqref{bp1IntRed}. 
To prove that $\psi(t)$ is a solution to nonlocal problem \eqref{SchrodEqt},\eqref{eq:linear_nc} we need to check if it satisfies \eqref{eq:linear_NCNS}. 
This is trivially true, since  \eqref{bp1IntRed} is transformed into \eqref{eq:linear_NCNS} via the direct manipulation with initial state \eqref{eq:initial_state_NCNS} using \eqref{eq:w_repr}:
\[
\psi_0 = \psi_1 - \psi_1 + \psi_0 =  \psi_1 - w = \psi_1-\suml_{k=1}^{n}\alpha_k \psi(t_k).
\]
The uniqueness of solution \eqref{bp1IntRed} to the given nonlocal problem  follows from the linear nature of both differential equation \eqref{SchrodEqt} and nonlocal condition \eqref{eq:linear_nc} as well as from the fact that  $\psi(t) \equiv 0$ when $\psi_1$ and $v(t)$ are equal to zero simultaneously. 
\end{proof}
We highlight, that the proof of sufficiency relies only on the assumptions needed for the existence of operator function $B(H)\psi$ for any $\psi \in X$. 
These assumptions do not include the requirement of $H$ having at least one eigenvector, that is essential to prove the necessity of Theorem \ref{thm:NCNS_exist_mild}. 
The theorem concerns the existence and uniqueness of the solution for any possible combination of $\psi_1$ and $v(t)$.
It does not discount the existence of solutions other than \eqref{bp1IntRed} for some specific combination of $\psi_1$ and $v(t)$. 
Namely, if the nonzero initial data $\psi_1, v(t)$ is chosen in such a way that the right-hand side of 
\begin{equation}\label{eq:nc_nonzero_v}
B\psi(0) = \psi_1 - \suml_{i=1}^n \alpha_i \intl_0^{t_i} U(t_i-s)v(s) ds
\end{equation}
is zero and there is a non-empty intersection between the set of roots of $b(z)$ and the spectrum of $H$, then  one can construct a whole family of non-trivial solutions to \eqref{SchrodEqt}, \eqref{eq:linear_nc}.
Because, as we have shown in the proof, every eigenstates of $H$ for which the corresponding eigenvalue coincides with the root of $b(z)$, will satisfy \eqref{eq:nc_nonzero_v} with the zero right-hand side.

It should also be noted, that by its structure, formula \eqref{bp1IntRed} resembles representation \eqref{eq:MildSolSchrodCP} of the solution to classical Cauchy problem \eqref{SchrodEqt}, \eqref{eq:SEIC}. 
More precisely, the following is true.
\begin{corollary}\label{thm:NCNS_equiv_Cauchy}
Assume that the requirements of Theorem \ref{thm:NCNS_exist_mild} are fulfilled, then the mild solution of nonlocal problem \eqref{SchrodEqt}, \eqref{eq:linear_nc} is equivalent to the solution of  classical Cauchy problem \eqref{SchrodEqt}, \eqref{eq:SEIC} represented by \eqref{eq:MildSolSchrodCP}, with the initial state $\psi_0$ defined by \eqref{eq:initial_state_NCNS}.
\end{corollary}

The correspondence between the solution of nonlocal problem and the solution of the classical Cauchy problem permits us to establish other important properties of \eqref{SchrodEqt}, \eqref{eq:linear_nc}.
\begin{corollary}\label{thm:NCNS_exist_strong}
Assume that in addition to the requirements of Theorem \ref{thm:NCNS_exist_mild} on $H$, both $b(z)$, $\psi_1$ belong to $D(H)$ 
and either one of two following conditions is satisfied:
	\begin{description}
		\item{a)} $v(t) \in D(H)$ and $v(t)$, $H v(t)$ are continuous on $[0, T]$, or 
		\item{b)} $v(t)$ is continuously differentiable on $[0, T]$.
	\end{description}
Then \eqref{bp1IntRed} is a strong (genuine) solution of nonlocal problem \eqref{SchrodEqt}, \eqref{eq:linear_nc}.
\end{corollary}
\begin{proof}
	We proceed by reducing the proof to the corresponding results on the genuine solution of the classical Cauchy problem \cite[Lemma 5.1]{Fattorini1985}. 
	In order to achieve that it is enough to show that the theorem's assumptions imply $\psi_0 \in D(H)$ or, which is equivalent, that $H\psi_0$ is well defined. 
	We depart from \eqref{eq:initial_state_NCNS} and use the above-mentioned properties of function calculus for strip-type operators:
	\[
		\begin{split}
			H\psi_0 = &HB^{-1}\left (\psi_1 - \suml_{i=1}^n \alpha_i \intl_0^{t_i} U(t_i-s)v(s) ds\right ) \\ 
		            = &B^{-1}H\psi_1 -B^{-1}\suml_{i=1}^n \alpha_i \intl_0^{t_i} U(t_i-s)Hv(s) ds.
		\end{split}
	\]
	The first term in the last formula is well defined because $\psi_1 \in D(H)$ and there always exists a sequence of states from $D(H^{\delta})$ with $\psi_1$ as a limit, such that $B^{-1}H$ is bounded 
	on the elements of that sequence. 
	By the same token we can show the well-definiteness of the second term, under assumption that \textit{a)} is true. 
	The case of \textit{b)}, as well as the rest of the proof, literally repeats the proof of the mentioned Lemma 5.1 from \cite{Fattorini1985}, and thus will be omitted here.  
%
\end{proof}
The conditions necessary for the existence of the strong solution are closely related to the well-posendess of \eqref{SchrodEqt}, \eqref{eq:linear_nc}. The evolution problem is called uniformly well-posed in $t \in [0, T]$ (see Section 1.2 of \cite{Fattorini1985}), if and only if the strong solution exists for a dense subspace of the initial data and the solution operator is uniformly bounded in $t$ on compact subsets of $[0, T]$. 

\begin{corollary}\label{thm:NCNS_well_posed}
	Let $H$ be an operator satisfying the assumptions of Theorem \ref{thm:NCNS_exist_mild}. 
	The nonlocal problem \eqref{SchrodEqt}, \eqref{eq:linear_nc} is uniformly well-posed in $t \in \R$ for any bounded $t_k\in [0, T]$, $\alpha_k \in \C$ if and only if all the zeros of $b(z)$ defined by \eqref{zerosExpI} are separated from $\Sigma$ .
\end{corollary}
\begin{proof}
	In Corollary \ref{thm:NCNS_exist_strong} we've already identified the dense subset $D(H)$ of $X$ such that for any $\psi_1 \in D(H)$ there exists a genuine solution of \eqref{SchrodEqt}, \eqref{eq:linear_nc}. 
	Assumptions on the parameters of nonlocal condition imply the boundedness of $B^{-1}$. 
	In Section \ref{sec:hs_operator_calculus} we mentioned that $U(t)$ is bounded as well, thus the solution operator from \eqref{bp1IntRed}  is bounded. 
	To conclude the proof we recall that the propagator $U(t)$ forms the group for $t \in \R$, hence the bounded solution operator is also uniformly bounded \cite[Theorem 2.1]{Fattorini1985}. 
\end{proof}

\begin{example}\label{ex1}
	Let us consider a two point version of nonlocal problem \eqref{SchrodEqt}, \eqref{eq:linear_nc}. 
	In such simple case, nonlocal condition \eqref{eq:linear_nc} takes the form
	\begin{equation}\label{ex:nc_two_point}
	u(0)+\alpha_1 u(t_1) =u_0,\quad t_1>0.
	\end{equation}
	Here we assume that $H$ has all the properties mentioned in Theorem \ref{thm:NCNS_exist_mild}.
	To determine the location of zeros of $b(z)$ we need to solve the equation
	$$
	1 +\alpha_1e^{-zi t_1} = 0,
	$$
	assuming that $\alpha_1 \in \C$ and $t_1 \in [0, T]$ are given. It has an infinite number of solutions $z_m$ 
	\begin{equation}
	\begin{split}\label{kerB1p}
	z_m = & -\frac{1}{it_1}\ln\left(-\frac{1}{\alpha_1}\right)=
	\\
	= &\frac{1}{t_1}\left[\Arg\left(\frac{1}{\alpha_1}\right)+2\pi m  +i\ln\left|\frac{1}{\alpha_1}\right|\right], \quad m\in \Z. 
	\end{split}
	\end{equation}
	Here $\Arg\left(\cdot\right)$ stands for a principal value of the argument of complex number. 
	The zeros $z_m$ are situated on the line, where the imaginary part  $\Im{z}= \ln{\left |1/\alpha_1\right |}/t_1$ is constant. They will belong to $\C\setminus\Sigma$  if  $\left |\Im{z}\right |$ is greater than the half-height $d$ of the spectrum $\Sigma$ defined by \eqref{eq:SpHalfStrip}. Consequently, the solution of \eqref{SchrodEqt}, \eqref{ex:nc_two_point} exists  if and only if 
	\begin{equation}\label{eq:NCN_est_nc}
		|\alpha_1| < e^{-t_1d},\quad \text{or}\quad 	|\alpha_1| > e^{t_1d}.\\
	\end{equation}
	The given nonlocal problem is well-defined for any $\alpha_1 \in \C$, except for the complex numbers  lying in the annulus $e^{-t_1d}\leq |\alpha_1| \leq e^{t_1d}$. 
	
\end{example}
It is important to note that constraints \eqref{eq:NCN_est_nc} enforce $|\alpha_1|\neq 1$. That requirement can be relaxed for some $\psi_1, v(t)$ if the spectrum of $H$ is disjoint in the neighborhood of $\R$. 
%
Another unique feature of the two-point problem \eqref{SchrodEqt}, \eqref{ex:nc_two_point} 
is expressed by one's ability to write the closed-form solution \eqref{kerB1p}, without specifying $\alpha_1$ beforehand. 
It becomes impossible for the general case of multi-point nonlocal condition \eqref{eq:linear_nc}, where one must rely on the numerical procedures to solve $b(z)=0$ and for that reason predefine the parameters of nonlocal condition. 
For many applications of \eqref{SchrodEqt},\eqref{eq:linear_nc} with $n>1$ this is not enough as one still would like to have some apriori information about the admissible set of $\alpha_k$ rather than simply check the existence of solution  for a fixed sequence $\alpha_k$, $k=\overline{1,n}$. 

\section{Zeros of $b(z)$ and associated problem for polynomials}\label{sec:NCNS_polynomial_reduction}
To find a way around the direct solution of $b(z)=0$, $n>1$ we start with a general observation suggested by Example \ref{ex1}.
The function $b(z)$ can be arbitrary closely approximated by a periodic function ${b^{\star}(z)}\equiv 1+\sum_{k=1}^n{\alpha_k e^{(-i {t^{\star}_k} z)}}$, where each ${t^{\star}_k}$ is the rational approximation to the corresponding real number $t_k$, $k=\overline{1,n}$. 
The function $b^{\star}(z)$ better suits our needs than $b(z)$, because the equation $b^{\star}(z)=0$ can always be reduced to the polynomial root finding problem. 

Let 
	\begin{equation*}
	t_k= \frac{\lambda_k}{\mu_k}, \quad \lambda_k \in \mathbb{Z},\quad  \mu_k \in \mathbb{N},
	\end{equation*}
we set $c_k=\frac{Q \lambda_k}{\mu_k}$, where 
\[
Q = \frac{\mbox{LCM}(\mu_1,\mu_2, \ldots, \mu_n)}{\mbox{GCD}(\lambda_1,\lambda_2, \ldots, \lambda_n)}
\]
is the ratio of the least common multiple (LCM) and the greatest common divisor (GCD) of the numerators and denominators of $t_k$ correspondingly.
A substitution  
\begin{equation}
\label{eq:NCNS_tranfs}
\Phi:\ u=\exp{\left (-iz/Q \right )}
\end{equation}
transforms the original problem about the location of zeros of $b(z)$ in $\C\setminus\Sigma$ into the problem about the location of zeros of a polynomial
\begin{equation}
\label{eq:NCNS_zeros_pol}
r(u)=1+\suml_{k=1}^{n}\alpha_k u^{c_k}
\end{equation}
in the exterior of an annulus
\[
\Upsilon:\quad e^{-d/Q}\leq |u| \leq e^{d/Q}, \quad u \in \C.
\]
The polynomial root finding problem for $r(u)=0$ is extensively studied (see \cite{Milovanovic2000a, Milovanovic2000}, as well as \cite{Sendov1994625,Henrici1974v1}). 
Polynomial $r(u)$ has exactly $c_n$ roots $u_k$ over $\C$. 
Their closed form representation exists for $c_n\leq 4$. 
So, now we technically can write the exact solvability conditions for \eqref{SchrodEqt}, \eqref{eq:linear_nc} in terms of $\alpha_k$ for $k$ up to 4.    
More importantly, it is possible to avoid the full solution of $r(u)=0$ altogether whilst checking $u_k \in \C\setminus\Upsilon$:
\begin{equation}\label{eq:NCNS_ext_annulus}
 |u_k|< e^{-d/Q} \lor |u_k| > e^{d/Q}, \quad k=\overline{1,c_n}.
\end{equation}
The shape of $\Upsilon$ suggests that we focus on a subclass of available root finding methods with results stated in the form of bounds \eqref{eq:NCNS_ext_annulus}.
Among those, we choose three effective complex root bounds (see \cite{Batra2016,Rump2003,PhdThSytnykEN} for the discussion and comparisons) for $P(u)=\suml_{k=0}^{N}a_k u^k$.
We have ordered them by the increasing computational complexity. 
Each of the following bounds has been reformulated as a double estimate to better fit \eqref{eq:NCNS_ext_annulus}.
	\begin{lemma}(\cite{Milovanovic2000a}, Theorem 2.4)\label{lem:zpol_2}
		The zeros of $P(u)$ satisfy the following inequalities:
	\[
	\begin{split}
			|u|\leq \left(1+\left(\frac{M_s}{|a_N|}\right)^q \right)^{1/q},\quad
			|u|\geq \frac{|\a_0|}{\left(|\a_0|+M_s^q \right)^{1/q}}, \\ 			
			M_s=\left(\sum\limits_{k=1}^N |\a_k|^s\right)^{1/s},\quad s,\ q \in \R_{>1},\quad \frac{1}{s}+\frac{1}{q}=1
\end{split}
\]
	\end{lemma}
{
The next estimate is due to M. Fujiwara \cite{Fujiwara1916}. 
It is the nearly optimal homogeneous bound in the space of polynomials \cite{Batra2016}:
\begin{lemma}\label{lem:zpol_3}
	All zeros of  $P(u)$ satisfy the inequalities
	\[
	\begin{split}
	|u| \leq 2\max\left \{\left |\frac{\a_0}{2\a_N}\right |^{1/N}, \left |\frac{\a_1}{\a_N}\right |^{1/(N-1)}, \ldots, \left |\frac{\a_{N-1}}{\a_N}\right |\right \},\\
	|u| \geq \frac{1}{2}\min\left \{\left |\frac{2\a_N}{\a_0}\right |^{1/N}, \left |\frac{\a_N}{\a_1}\right |^{1/(N-1)}, \ldots, \left |\frac{\a_N}{\a_{N-1}}\right |\right \},
	\end{split}
	\]
	where $1/0=+\infty$.
\end{lemma}
}  
The third estimate, originally proved by H.~Linden \cite{Linden1998}, gives bounds on the real and imaginary part of zeros separately. It has been adapted in \cite{PhdThSytnykEN} to fit within the framework studied here. 
\begin{lemma}\label{lem:zpol_4}
	All  zeros of $P(u)$ satisfy the double estimate $\max\{V_1^{-1},V_2^{-1}\} \leq |u|\leq \min\{V'_1,V'_2\}$, where
	$$
	V_1=\cos{\frac{\pi}{N+1}}+\frac{|\alpha_N|}{2|\alpha_0|}\left(\left|\frac{\alpha_1}{\alpha_N}\right|+\sqrt{1+\suml_{k=1}^{N-1}\left|\frac{\alpha_k}{\alpha_N}\right|^2}\right)
	$$
	\begin{multline*}
		V_2=\frac{1}{2}\left(\left|\frac{\alpha_1}{\alpha_0}\right|+\cos\frac{\pi}{N}\right)\\
		+\frac{1}{2}\left[\left(\left|\frac{\alpha_1}{\alpha_0}\right|-\cos\frac{\pi}{N}\right)^2+
		\left(1+\left|\frac{\alpha_N}{\alpha_0}\right|\sqrt{1+\suml_{k=2}^{N-1}\left|\frac{\alpha_k}{\alpha_N}\right|^2}\right)^2\right]^{1/2}
	\end{multline*}
	and $V'_i$ is  obtained from $V_i$ by the substitution $\a_k =\a_{N-k}$, $k=\overline{0,N}$, $i=1,2$.
\end{lemma}
Now we are in the position to formulate our next result.
\begin{theorem}\label{thm:NCNS_suf_cond_root_est}
Suppose that operator $H$ from \eqref{SchrodEqt} satisfies the assumptions of Theorem \ref{thm:NCNS_exist_mild} and
all $t_k$ in \eqref{eq:linear_nc} are rational numbers.
If at least one bound from Lemmas \ref{lem:zpol_2} - \ref{lem:zpol_4} for polynomial \eqref{eq:NCNS_zeros_pol} induce \eqref{eq:NCNS_ext_annulus}, then the nonlocal problem \eqref{SchrodEqt}, \eqref{eq:linear_nc} has the following properties:
\begin{enumerate}
	\item \label{NCNS_prop1} it is uniformly well-posed in $t \in \R$;
	\item for any $\psi_1 \in X$, $v \in L^1((0;T),X)$ there exists mild solution \eqref{bp1IntRed} with the characteristics mentioned in Theorem \ref{thm:NCNS_exist_mild};
	\item \label{NCNS_prop3}solution \eqref{bp1IntRed} will also be strong if $\psi,v(t)$ satisfy either of the requirements a) or b) from Corollary \ref{thm:NCNS_exist_strong}.
\end{enumerate}
\end{theorem}
\begin{proof}
If the zeros $u_k$ of \eqref{eq:NCNS_zeros_pol} obey \eqref{eq:NCNS_ext_annulus}, their images
\[
	z_k = \Phi^{-1}(u_k) = 	Q\left[\Arg\left(u_k\right)+2\pi m  +i\ln\left|u_k\right|\right],  	
\]
are clearly in the interior of $\C\setminus\Sigma$ no matter what is the value of $m\in \Z$.
The application of Theorem \ref{thm:NCNS_exist_mild} and Corollaries  \ref{thm:NCNS_exist_strong},\ref{thm:NCNS_well_posed}
concludes the proof.
\end{proof}


 The result of Theorem \ref{thm:NCNS_suf_cond_root_est} can be turned into criteria by enforcing the necessary and sufficient conditions 
 for the validity of \eqref{eq:NCNS_ext_annulus}  derived via the Schur-Cohn algorithm \cite[p. 493]{Henrici1974v1}. 
 For a given polynomial $r(u)$ the algorithm produces 
 a set of up to $2 c_n$ inequalities, that are polynomial in $\alpha_k$, $k=\overline{1,n}$. These inequalities need to be valid simultaneously in order for the Schur-Cohn test to pass \cite[Thm. 6.8b]{Henrici1974v1}. The precise result is stated as follows.
 \begin{theorem}\label{thm:NCNS_criteria}
 Suppose that operator $H$ from \eqref{SchrodEqt} satisfies the assumptions of Theorem \ref{thm:NCNS_exist_mild} and
 all $t_k$ in \eqref{eq:linear_nc} are rational numbers.
 Nonlocal problem \eqref{SchrodEqt}, \eqref{eq:linear_nc} has properties \ref{NCNS_prop1}--\ref{NCNS_prop3} of Theorem \ref{thm:NCNS_suf_cond_root_est} if and only if  
 the polynomials $b(e^{d/Q}u)$, $u^{c_n} b(e^{-d/Q}u)$ pass the Schur-Cohn test for the given set of parameters $\alpha_k \in \C$, $k=\overline{1,n}$ from \eqref{eq:linear_nc}. 
 \end{theorem}
 \begin{proof}
 	The substitution $u = e^{d/Q}u'$ ($u = u'^{c_n} b(e^{-d/Q}u')$) transforms right (left) inequality from \eqref{eq:NCNS_ext_annulus} into the inequality $u'_k > 1$. 
 	In both cases the validity of the last inequality is checked by the Schur-Cohn test \cite[Thm. 6.8b]{Henrici1974v1}. 
 	"If" clause of Theorem \ref{thm:NCNS_exist_mild} along with Corollaries  \ref{thm:NCNS_exist_strong},\ref{thm:NCNS_well_posed} assures the sufficiency. 
 	Mapping $\Phi$ is a bijection of the vertical strip $|\Im{z}|\leq \pi \frac{\mbox{GCD}(\mu_1,\mu_2, \ldots, \mu_n)}{\mbox{LCM}(\lambda_1,\lambda_2, \ldots, \lambda_n)}$ onto $\C$. The strip's height equals to the period of $b(z)$. This fact guaranties the necessity via application of the "only if" clause of Theorem \ref{thm:NCNS_exist_mild} and Corollaries  \ref{thm:NCNS_exist_strong},\ref{thm:NCNS_well_posed}. 
 \end{proof}

It remains to study the following question:  what happens when some of $t_k$ are irrational?
Consider an approximation $b^{\star}(z)$ of $b(z)$ mentioned above.  If $t^{\star}_k \rightarrow t_k$, $k=\overline{1,n}$ the function $b^{\star}(z)$ uniformly converges to $b(z)$ on the compact subsets of the open set containing $\Sigma$. 
Hurwitz theorem \cite[Corollary 4.10f]{Henrici1974v1} provides the means to claim that all zeros of $b(z)$ lies in the interior of $\C\setminus\Sigma$, if that is true for $b^{\star}(z): t^{\star}_k \rightarrow t_k$. 
The degree of a polynomial $r^{\star}(u)$ corresponding to $b^{\star}(z)$ grows to $\infty$ when $t^{\star}_k \rightarrow t_k$ and this $t_k$ is irrational. 
But, its coefficients $\alpha_k$ are not affected by the increase of $c_n^{\star}$.
This keeps the root estimates from Lemmas \ref{lem:zpol_2}--\ref{lem:zpol_4} meaningful. 
As a result we have arrived at the following corollary.
\begin{corollary}
Assume that for every $k=1,\ldots, n$ the sequence of rational numbers $\left \{t^{\star}_{kl}\right \}_{l=1}^\infty$, is such that $\lim\limits_{l \rightarrow \infty} t^{\star}_{kl}=t_k$. If the conditions of Theorem \ref{thm:NCNS_suf_cond_root_est} regarding the roots of $r^{\star}_l(u)$ associated with $t^{\star}_{kl}$, $k=\overline{1,n}$ are fulfilled for all $l>0$, then the rest of theorem's statement remains valid for $t_k \in \R$.
\end{corollary}

Now, we would like to compare the conditions on $\alpha_{1},\alpha_{2} \in \R$ obtained with help of Theorems \ref{thm:NCNS_suf_cond_root_est}, \ref{thm:NCNS_criteria} against the previously known condition \eqref{estLiang2002}. 
\begin{example}
	Let us consider a three point nonlocal condition
	\begin{equation}
		u(0)+\alpha_1 u(t_1)+\alpha_2 u(t_2)=0,\quad t_1,t_2 >0.
	\end{equation}
	For simplicity we set $t_1=1,\ t_2=2$ and consider the non-zero spectral half-height $d=\pi/40$.
	Then, the equation $b(z)=0$ is reduced to $1+\alpha_1 u + \alpha_2 u^2=0$. 
	As shown in Fig. \ref{fig:NCNS_est_comp} b), the exact conditions on $\alpha_{1},\alpha_{2}$ calculated by Theorem  \ref{thm:NCNS_criteria} (the Schur-Cohn algorithm)
	\begin{equation}\label{eqintshura2}
	\left[
	\begin{array}{l}
	\begin{cases}
	|a_{{2}}|^{2} < \e^{-4d},\\
	\e^{4d}|a_{{1}}|^{2}|a_{{2}}|^{2}-\e^{6d}|a_{{2}}|^{4}-2 \e^{4d}|a_{{2}}|\left (|a_{{1}}|^{2}-|a_{
			{2}}|\right) +|a_{{1}}|^{2} < \e^{-2d}
	\end{cases} 
	 \\
	\begin{cases}
	|a_{{2}}|^{2} > \e^{4d},\\
	\e^{-4d}|a_{{1}}|^{2}|a_{{2}}|^{2}-\e^{-6d}|a_{{2}}|^{4}-2 \e^{-2d}|a_{{2}}|\left (|a_{{1}}|^{2}-|a_{
		{2}}|\right) +|a_{{1}}|^{2} > \e^{2d}
	\end{cases}
	\end{array}
	\right.
	\end{equation}
	lead to a considerably wider class of admissible pairs $\left (\alpha_{1},\alpha_{2}\right )$ than those obtained by \eqref{estLiang2002}. 
	\begin{figure}[hbt]%
		\begin{tabular}{c}
			\begin{overpic}[width=0.46\linewidth]%
				{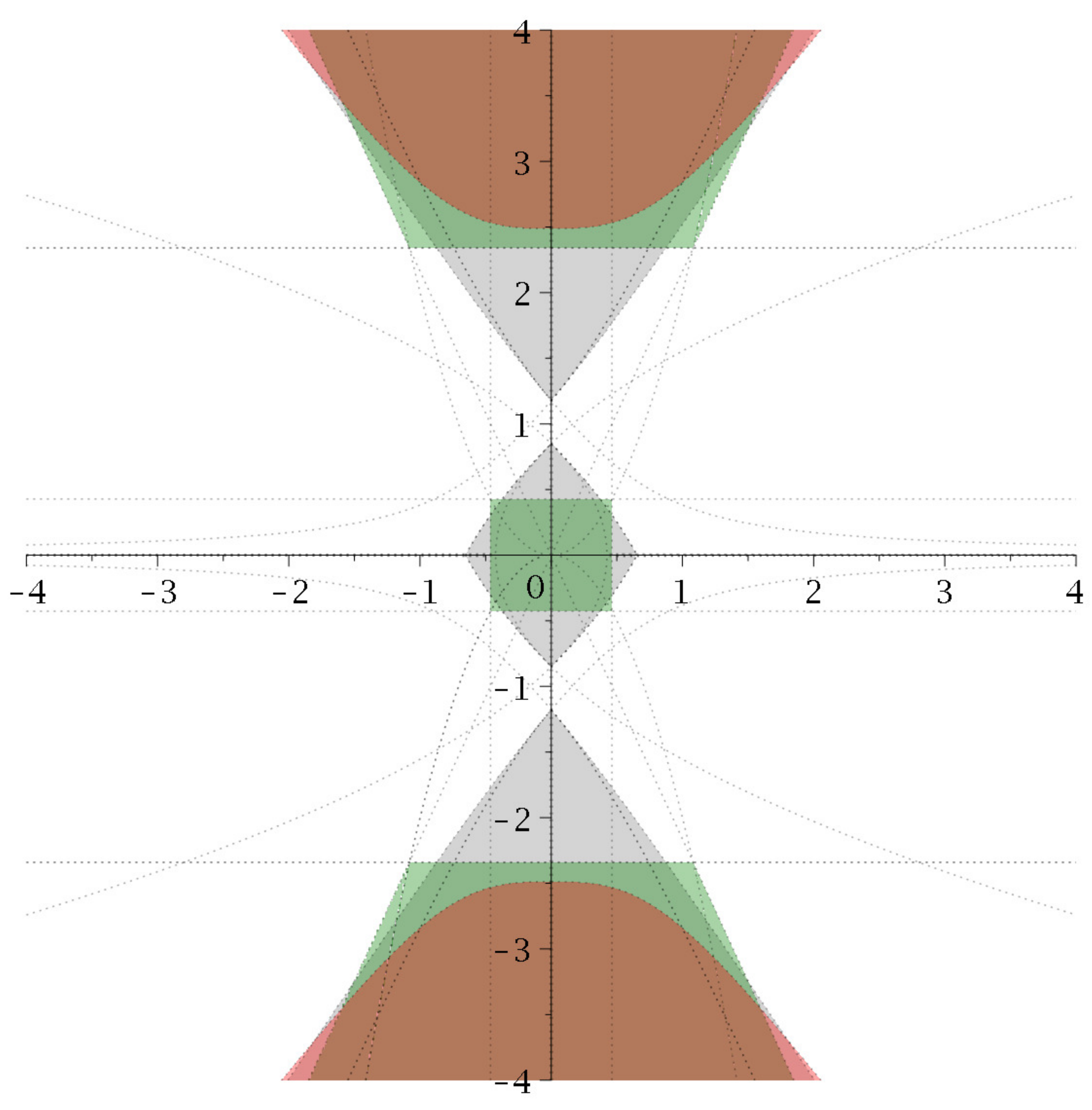}
				\put(5,5){\footnotesize\textbf{a)}}
				\put(90,52){\footnotesize $\alpha_1$}
				\put(53,90){\footnotesize $\alpha_2$}
			\end{overpic}
			\hfill
			\begin{overpic}[width=0.46\linewidth]%
				{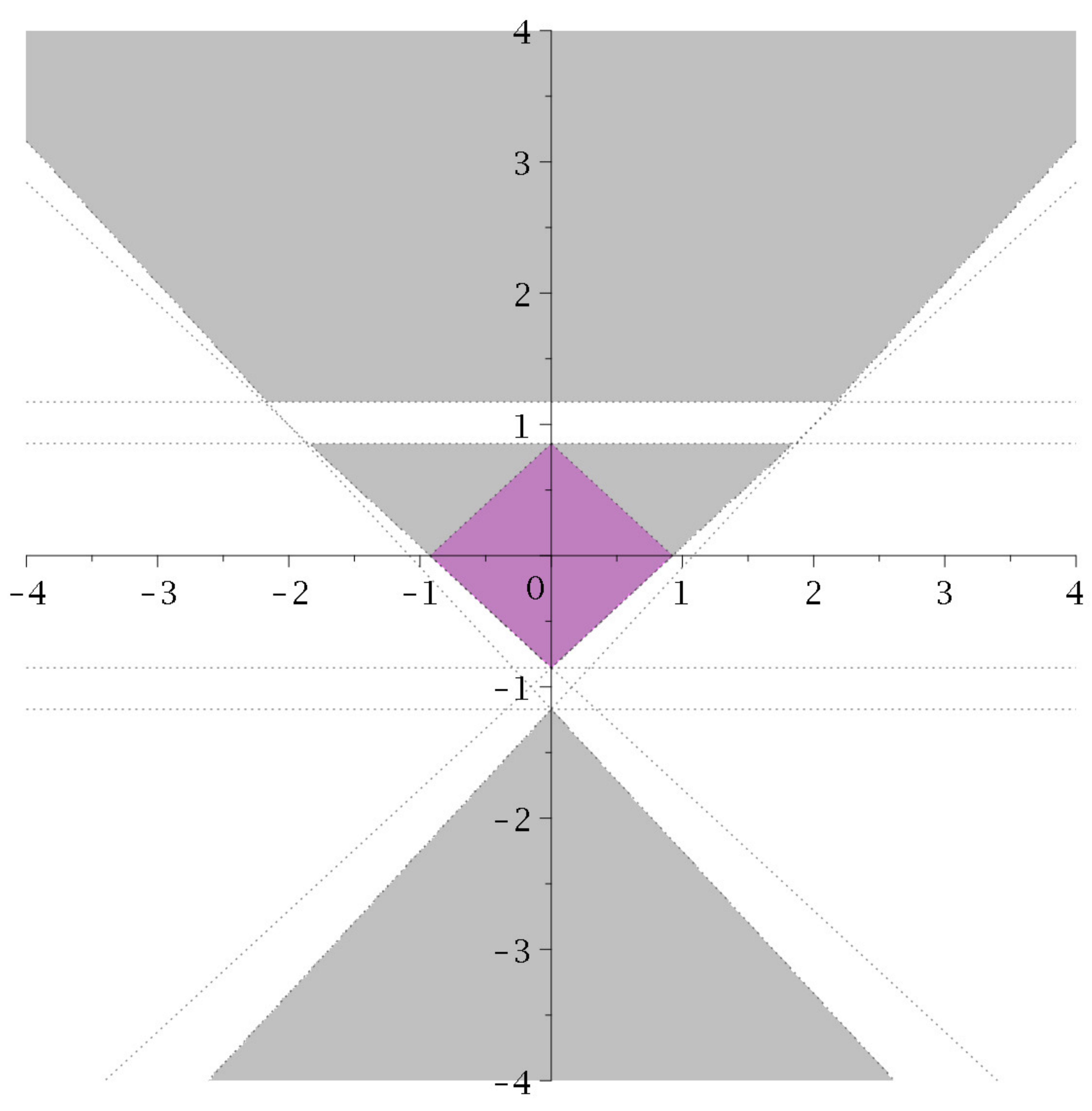}
				\put(5,5){\footnotesize\textbf{b)}}
				\put(90,52){\footnotesize $\alpha_1$}
				\put(53,90){\footnotesize $\alpha_2$}
			\end{overpic}
		\end{tabular}
		\caption{The regions (filled) in the space of parameters $\alpha_{1},\alpha_{2} \in \R$ from \eqref{eq:linear_nc} where problem \eqref{SchrodEqt},\eqref{eq:linear_nc} is well-posed, $d=\pi/40$, $t_1=1, t_2=2$  (color online). \newline
			\textbf{a)} Application of Theorem \ref{thm:NCNS_suf_cond_root_est} and root estimates from:  Lemma \ref{lem:zpol_2} with $s=q=2$ -- dark gray (red); Lemma \ref{lem:zpol_3} -- middle gray (green); Lemma \ref{lem:zpol_4} -- light gray; \newline
			\textbf{b)} The complete set of feasible $\left (\alpha_{1},\alpha_{2}\right )$ via the application of Theorem \ref{thm:NCNS_criteria} -- gray, and set of pairs based on the estimate \eqref{estLiang2002} -- dark grey (violet). }
		\label{fig:NCNS_est_comp}
	\end{figure}
	In fact, the second system of inequalities from \eqref{eqintshura2} gives rise to the unbounded region (union of two unbounded sets depicted in Fig. \ref{fig:NCNS_est_comp} b) in the space of parameters $\alpha_{1},\alpha_{2} \in \R$, meanwhile the solutions of \eqref{estLiang2002} are strictly bounded in $|\alpha_1|, |\alpha_2|$ (the interior of the rhombic region in Fig. \ref{fig:NCNS_est_comp} b). They lay within the isosceles triangle which acts as graphical solution of the first system of inequalities in \eqref{eqintshura2}.   
	The gap between this triangle and the two other regions containing the solutions of \eqref{eqintshura2} shortens when $d\rightarrow 0$, and in the limit is described by $|\alpha_2|=1$. 
	Comparison of generalized condition \eqref{estLiang2002} from \cite{ashyralyev2008nonlocal,Byszewski1992} and the sufficient conditions provided by Theorem \ref{thm:NCNS_suf_cond_root_est} (depicted in Fig. \ref{fig:NCNS_est_comp} a) unveils that \eqref{estLiang2002} performs better than the inner circle estimates  of Lemmas \ref{lem:zpol_2}--\ref{lem:zpol_4} (the part $\C \setminus \Upsilon$ defined by the first inequality in \eqref{eq:NCNS_ext_annulus}). 
	Therefore, when it comes to the apriori estimates on the parameters of nonlocal condition, we advice to use the combination of \eqref{estLiang2002} and the part of Theorem \ref{thm:NCNS_suf_cond_root_est} which implies $|u_k| > e^{d/Q}$. 
\end{example}
\section*{Conclusions}
We established exact dependence of the solution to problem \eqref{SchrodEqt}, \eqref{eq:linear_nc} on the parameters of nonlocal condition, derived the well-posedness criteria, and proved the theorems regarding existence of the problem's mild (strong) solution.  
The conditions on existence of the solution to the given nonlocal problem obtained here, generalize other available results \cite{ashyralyev2008nonlocal,NonlocalAbsNonLinNtouyas1997} beyond the case of $\alpha_k$ bounded by \eqref{estLiang2002}. 

To illustrate the applicability of the obtained results we use them to study the existence of solutions to the general two- and three- point nonlocal problems. 
For each model problem we were able to describe analytically the entire manifold of admissible parameters of the corresponding nonlocal condition.

Aside from the case when $H$ is a classical Schr\"odinger operator, our method of analysis covers the situation when $H$ in \eqref{SchrodEqt} is non-Hermitian. 
The later situation occurs in the modeling of open-quantum systems, where the anti-Hermitian part of Hamiltonian describes the interaction of quantum system with the environment (see \cite{Zloshchastiev2014} and the references therein).

The technique used to prove the main result of this paper relies on the linear nature of the problem and the existence of exact representation for the solution operator via the Dunford-Cauchy formula.
The results of the paper, therefore, can be generalized to other linear nonlocal problems for Schr\"odinger equation.
This will be the subject of our future work.
{\small
\bibliographystyle{siam}


\def\bibpath{.}

\bibliography{%
\bibpath/sytnyk,%
\bibpath/nonlocal,%
\bibpath/nonlocal_schrodinger,%
\bibpath/matan,%
\bibpath/root_finding,%
\bibpath/abstract_schrodinger,%
\bibpath/operator_calculus,%
\bibpath/driven_quantum_systems,%
}
}
\end{document}